\documentclass[runningheads]{llncs}
\usepackage{color}
\definecolor{keywordcolor}{rgb}{0.7, 0.1, 0.1}   
\definecolor{tacticcolor}{rgb}{0.0, 0.1, 0.6}    
\definecolor{commentcolor}{rgb}{0.4, 0.4, 0.4}   
\definecolor{symbolcolor}{rgb}{0.0, 0.1, 0.6}    
\definecolor{sortcolor}{rgb}{0.1, 0.5, 0.1}      
\definecolor{attributecolor}{rgb}{0.7, 0.1, 0.1} 

\usepackage{listings}

\usepackage{algorithmicx}
\usepackage[noend]{algpseudocode}
\usepackage[T1]{fontenc}
\usepackage{amssymb, upgreek}
\usepackage{stmaryrd}
\usepackage{bussproofs}
\usepackage{hyperref}
\usepackage{todonotes}
\usepackage{unixode}
\usepackage{cite}

\newtheorem{prop}{Proposition}

\newcommand{\ttt}[1]{\texttt{\small #1}}

\newcommand{\PI}{\Uppi}
\newcommand{\lam}{\lambda}

\begin{document}

\pagestyle{headings}  
\title{Congruence Closure in Intensional Type Theory}
\titlerunning{Congruence Closure in Intensional Type Theory}
\author{Daniel Selsam\inst{1} \and Leonardo de Moura\inst{2}}
\authorrunning{Selsam and de Moura} 

\institute{Stanford University\\
  \email{dselsam@stanford.edu}
  \and
  Microsoft Research\\
  \email{leonardo@microsoft.com}}
\maketitle

\begin{abstract}
  Congruence closure procedures are used extensively in automated
  reasoning and are a core component of most satisfiability
  modulo theories solvers.
  However, no known congruence closure algorithms can support
  any of the expressive logics based on intensional type theory (ITT), which
  form the basis of many interactive theorem provers.
  The main source of expressiveness in these logics is dependent
  types, and yet existing congruence closure procedures found in
  interactive theorem provers based on ITT do not handle dependent
  types at all and only work on the simply-typed subsets of the
  logics.
  Here we present an efficient and proof-producing congruence closure procedure
  that applies to every function in ITT no matter how many
  dependencies exist among its arguments, and that only relies on the commonly assumed \emph{uniqueness of identity proofs} axiom.
  We demonstrate its
  usefulness by solving interesting verification problems involving
  functions with dependent types.
\end{abstract}

\section{Introduction}

Congruence closure procedures are used extensively in automated reasoning, since
almost all proofs in both program verification and formalized mathematics
require reasoning about equalities~\cite{nelson1980fast}. The algorithm
constitutes a fundamental component of most satisfiability modulo theories (SMT) solvers~\cite{barrett2011cvc4,de2008z3}; it is
often distinguished as the ``core theory solver'', and is responsible for communicating
literal assignments to the underlying SAT solver and equalities to the other
``satellite solvers''~\cite{ematching,de2008z3}. However, no known congruence
closure algorithms can support any of the expressive logics based on intensional type theory (ITT).
Yet despite the lack of an algorithm for congruence closure, the
benefits that ITTs confer in terms of expressiveness, elegance, and
trustworthiness have proved substantial enough that different flavors of ITT
form the basis of many interactive theorem provers, such as Coq~\cite{coq}, Lean~\cite{delean}, and Matita~\cite{matita}, and also several
emerging programming languages, such as Agda~\cite{agda}, Epigram~\cite{epigram}, and Idris~\cite{idris}.
Many of the most striking successes in both certified programming and formalized
mathematics have been in variants of ITT, such as the development of a fully-certified
compiler for most of the C language~\cite{leroy2009formal} and the
formalization of the odd-order theorem~\cite{gonthier2013machine}.

There are currently two main workarounds for the lack of a congruence closure
algorithm for ITT, and for the lack of robust theorem
proving tools for ITT more generally. One option is to rely much more on manual
proving. Although many impressive projects have been formalized with little to
no automation, this approach is not very attractive since the cost of manual
proving can be tremendous. We believe that as long as extensive manual proving is a central part
of writing certified software or formalizing mathematics, these will remain
niche activities for the rare expert. The other option is to
relinquish the use of dependent types whenever manual reasoning
becomes too burdensome so that more traditional automation can be
used. Note that the Coq system even has a tactic
\lstinline{congruence} that
performs congruence closure, but it does not handle dependent types at
all and only works on the simply-typed subset of the language.
This sacrifice may be appropriate in certain contexts, but losing all the benefits of
dependent types makes this an unsatisfactory solution in general.

Given the limitations of these two workarounds, it would be preferable
to perform congruence closure and other types of automated reasoning
directly in
the richer language of ITT. Unfortunately, equality and congruence are
both surprisingly subtle in ITT, and as we will see, the theorem that could justify
using the standard congruence closure procedure for functions with
dependent types is not provable in the core logic, nor does it follow from any of the axioms commonly assumed in existing systems.
In this paper, we introduce a new
notion of congruence that applies to every function in ITT no matter
how many dependencies exist among its arguments, along with a simple
and efficient extension of the standard congruence closure procedure
to fully automate reasoning about this more general notion of
congruence. Our procedure is applicable to a wide variety of projects since
it only relies on the \emph{uniqueness of identity proofs}
axiom, which is built into the logic of many systems including
Agda, Idris, and Lean, and which is commonly assumed in the
others. We hope our procedure helps make it possible for users to have the best of both worlds: to reap all the benefits of dependent types while still
enjoying all the power of traditional automation.

\section{Preliminaries\label{sec:prelim}}

We assume the term language is a dependent $\lam$-calculus in which terms are
described by the following grammar:
\begin{lstlisting}
                            $t$,$s$ ::= $x$ $\mid$ $c$ $\mid$ Type $\mid$ $t$ $s$ $\mid$ $\lam x : s, t$ $\mid$ $\PI x : s, t$
\end{lstlisting}
where $x$ is a variable and $c$ is a constant. To simplify the presentation,
we omit type universes at sort \lstinline{Type}. It is not relevant
to this paper whether the universe hierarchy is cumulative or not, nor whether
there is a distinguished sort \lstinline{Prop} (the sort of all propositions).
The term \lstinline{$\PI$x:A, B} denotes the type of functions \lstinline{f}
that map any element \lstinline{a:A} to an element of
\lstinline{B[a/x]}. When \lstinline{x} appears in \lstinline{B} we say
that \lstinline{f} is \emph{dependently-typed}; otherwise we write \lstinline{$\PI$x:A, B} as
\lstinline{A → B} to denote the usual non-dependent function
space. When \lstinline{B} is a proposition, \lstinline{$\PI$x:A, B} can be read as the universally quantified formula \lstinline{$\forall$x:A, B}, or as the logical implication \lstinline{A $\Rightarrow$ B} if \lstinline{x} does not appear in \lstinline{B}. The term \lstinline{f a} denotes a function application, and
the lambda abstraction \lstinline{$\lam$x:A, t} denotes a function that given an element \lstinline{a}
of type \lstinline{A} produces \lstinline{t[a/x]}.
As usual in Type Theory, a \emph{context} $\Gamma$ is a sequence of \emph{typing assumptions}
\lstinline{a:A} and (local) definitions \lstinline{c:A := t}, where \lstinline{t} has type \lstinline{A} and \lstinline{c} does not occur in \lstinline{t}.
We often omit the type \lstinline{A} and simply write \lstinline{c := t} to save space when no confusion arises.
Similarly, an \emph{environment} $\Delta$ is a sequence of (global)
definitions \lstinline{f:A := t}.
We use $type(\Delta, \Gamma, t)$ to denote the type of $t$ with
respect to $\Delta$ and $\Gamma$, and $type(t)$ when no confusion
arises. Given an environment $\Delta$ and a context $\Gamma$, every
term reduces to a normal form by the standard $\beta\delta\eta\iota\zeta$-reduction rules.
For this paper we will assume a
fixed environment $\Delta$ that contains all definitions and theorems
that we present.
As usual, we
write \lstinline{$\PI$(a:A)(b:B),C} as a shorthand for
\lstinline{$\PI$a:A,($\PI$b:B,C)}. We use a similar shorthand for
$\lam$-terms.

\subsection{Equality\label{sec:equality}}
One of the reasons that congruence is subtle in ITT is that
equality itself is subtle in ITT. The single notion of equality in
most other logics splits into at least three different yet related
notions in ITT.

\paragraph{Definitional equality.}
The first notion of equality in ITT is \emph{definitional equality}. We write
\lstinline{a $\equiv$ b} to mean that \lstinline{a} and \lstinline{b} are equal by
definition, which is the case whenever \lstinline{a} and \lstinline{b} reduce to
the same normal form. For example, if we define a function
\lstinline{f : $\mathbb{N}$ $\rightarrow$ $\mathbb{N}$ := $\lam$ n : $\mathbb{N}$, 0} in the
environment $\Delta$, then the terms \lstinline{0} and \lstinline{f 0} both reduce
to the same normal form \lstinline{0} and so are equal by definition.
On the other hand, \lstinline{($\lam$ n m: $\mathbb{N}$, n + m)} is not definitionally equal to \lstinline{($\lam$ n m: $\mathbb{N}$, m + n)}, since they
are both in normal form and these normal forms are not the same.
Note that definitional equality is a judgment at the meta-level, and the theory itself
cannot refer to it; in particular, it is not possible to assume or
negate a definitional equality.

\paragraph{Homogeneous propositional equality.}
The second notion of equality in ITT is \emph{homogeneous propositional
  equality}, which we will usually shorten to
\emph{homogeneous equality} since ``propositional'' is implied.
  Unlike definitional equality which is a
judgment at the meta-level, homogeneous equality can be assumed, negated, and
proved inside the logic itself. There is a constant
\lstinline[breaklines=false]{eq : $\PI$ (A : Type), A $\rightarrow$ A $\rightarrow$ Type} in $\Delta$
such that, for any type \lstinline{A} and elements \lstinline{a b : A}, the expression
\lstinline{eq A a b} represents the proposition that \lstinline{a} and \lstinline{b} are
``equal''. Note that we call this homogeneous equality because the types of
\lstinline{a} and \lstinline{b} must be definitionally equal to even
\emph{state} the proposition that \ttt{a} and \ttt{b} are equal.
We write \lstinline{a =$_\ttt{A}$ b}
as shorthand for \lstinline{eq A a b}, or \lstinline{a = b} if the type \lstinline{A} is clear
from context. We say a term \lstinline{t} of type \lstinline{a = b} is a \emph{proof} for
\lstinline{a = b}.

The meaning of homogeneous equality is given by the introduction and
elimination rules for \lstinline{eq}, which state how to prove that two elements
are equal and what one can do with such a proof respectively. The introduction rule
for \lstinline{eq} is the dependent function
\lstinline[breaklines=false]{refl : $\PI$ (A : Type) (a : A), a = a},
which says that every element of type \lstinline{A} is equal
to itself. We call \lstinline{refl} the reflexivity axiom, and write \lstinline{refl a}
whenever the type \lstinline{A} is clear from context. Note that if \lstinline{a b : A} are
definitionally equal, then \lstinline{refl a} is a proof for \lstinline{a = b}.
The elimination principle (also known as the recursor) for the type \lstinline{eq} is the dependent function \lstinline{erec}:
\begin{lstlisting}
erec : $\PI$ (A : Type) (a : A) (C : A $\rightarrow$ Type), C a $\rightarrow$ $\PI$ (b : A), a = b $\rightarrow$ C b
\end{lstlisting}
This principle states that if a property \lstinline{C} holds for an element
\lstinline{a}, and \lstinline{a = b} for some \lstinline{b}, then we can conclude that
\lstinline{C} must hold of \lstinline{b} as well. We say \lstinline{C} is the
\emph{motive}, and we write \lstinline{(erec C p e)} instead of \lstinline{(erec A a C p b e)}
since \lstinline{A}, \lstinline{a} and \lstinline{b} can be inferred easily from \lstinline{e : a = b}.
Note that by setting \lstinline{C} to be the identity function \lstinline{id : Type $\rightarrow$ Type},
\lstinline{erec} can be used to change the type of a term to an
equal type; that is, given a term \lstinline{a : A} and a proof
\lstinline{e : A = B}, the term \lstinline[breaklines=false]{(erec id a e)} has type
\lstinline{B}.
We call this a \emph{cast}, and say that we \emph{cast} \lstinline{a} to have type \lstinline{B}.
Note that it is straightforward to use \lstinline{erec} and \lstinline{refl} to
prove that \lstinline{eq} is symmetric and transitive and hence an
equivalence relation.

\paragraph{Heterogeneous propositional equality.}
As we saw above, homogeneous equality suffers from a peculiar
limitation: it is not even possible to form the proposition \lstinline{a = b} unless
the types of \lstinline{a} and \lstinline{b} are definitionally equal.
The further one strays from the familiar
confines of simple type theory, the more severe this handicap becomes. For
example, a common use of dependent types is to include the length of a list
inside its type in order to make out-of-bounds errors impossible. The
resulting type is often called a \emph{vector} and has type
\lstinline{vector : $\PI$ (A : Type), $\mathbb{N}$ → Type}.
It is easy to define an append function on vectors:
\begin{lstlisting}
app : $\PI$ (A : Type) (n m : $\mathbb{N}$), vector A n → vector A m → vector A (n + m)
\end{lstlisting}
However, we cannot even state the proposition that \lstinline{app} is associative
using homogeneous equality, since the type \lstinline{vector A (n + (m + k))} is not
definitionally equal to the type \lstinline{vector A ((n + m) + k)}, only
propositionally equal. The same issue arises when reasoning about
vectors in mathematics. For example, we cannot even state the
proposition that concatenating zero-vectors of different lengths $m$
and $n$ over the real numbers $\mathbb{R}$ is commutative, since the type
$\mathbb{R}^{m + n}$ is not definitionally equal to the type
$\mathbb{R}^{n + m}$. In both cases, we could use \lstinline{erec} to
cast one of the two terms to have the type of the other, but
this approach would quickly become unwieldy as the number of
dependencies increased, and moreover every procedure that reasoned
about equality would need to do so modulo casts.

Thus there is a need for a third notion of equality in ITT,
\emph{heterogeneous propositional equality}, which we will usually shorten to
\emph{heterogeneous equality} since ``propositional'' is implied. There is a constant
\lstinline{heq : $\PI$ (A : Type) (B : Type), A $\rightarrow$ B $\rightarrow$ Type}
that behaves like \lstinline{eq} except that its arguments may have different types.\footnote{There are many equivalent ways of defining \lstinline{heq}. One popular way is ``John Major equality''~\cite{mcbride2000elimination}. Additional formulations and formal proofs of equivalence can be found at {\scriptsize\url{http://leanprover.github.io/ijcar16/heq.lean}}.}
We write \lstinline{a == b} as shorthand for \lstinline{heq A B a b}.
Heterogeneous equality has an introduction rule
\lstinline{hrefl : $\PI$ (A : Type) (a : A), a == a} analogous to
\lstinline{refl}, and it is straightforward to show that
\lstinline{heq} is an equivalence relation by proving the
following theorems:
\begin{lstlisting}
hsymm : $\PI$ (A B : Type) (a : A) (b : B), a == b $\rightarrow$ b == a
htrans : $\PI$ (A B C : Type) (a : A) (b : B) (c : C), a == b $\rightarrow$ b == c $\rightarrow$ a == c
\end{lstlisting}
Unfortunately, the flexibility of \lstinline{heq} does not come without a cost: as we
discuss in \S\ref{sec:congruence}, \lstinline{heq} turns out to be
weaker than \lstinline{eq} in subtle ways and does not permit as
simple a notion of congruence.
\paragraph{Converting from heterogeneous equality to homogeneous equality.}
It is straightforward to convert a proof of homogeneous equality \lstinline{p : a = b} into one
of heterogeneous equality using the lemma
\begin{lstlisting}
lemma ofeq (A : Type) (a b : A) : a = b $\rightarrow$ a == b
\end{lstlisting}
However, we must assume an axiom in order to prove the reverse direction
\begin{lstlisting}
ofheq (A : Type) (a b : A) : a == b $\rightarrow$ a = b
\end{lstlisting}
The statement is equivalent to the
\emph{uniqueness of identity proofs} (UIP) principle~\cite{streicher:93}, to Streicher’s \emph{Axiom
K}~\cite{streicher:93}, and to a few other variants as well. Although these axioms
are not part of the core logic of ITT, they have been found to be consistent
with ITT by means of a meta-theoretic argument~\cite{miquel2003not}, and
are built into the logic of many systems including Agda, Idris, and Lean. They also follow from various stronger axioms that are commonly
assumed, such as \emph{proof irrelevance} and \emph{excluded
  middle}. In Coq, UIP or an axiom that implies it is often assumed when heterogeneous equality is
used, including in the CompCert project~\cite{leroy2009formal}. Our approach is built upon
being able to recover homogeneous equalities from heterogeneous
equalities between two terms of the same type and so makes heavy use
of \lstinline{ofheq}.

\section{Congruence\label{sec:congruence}}

\paragraph{Congruence over homogeneous equality.}
It is straightforward to prove the following lemma using \lstinline{erec}:
\begin{lstlisting}
  lemma congr : $\PI$ (A B : Type) (f g : A → B) (a b : A), f = g $\rightarrow$ a = b $\rightarrow$ f a = g b
\end{lstlisting}
and thus prove that \lstinline{eq} is indeed a congruence relation for simply-typed
functions. Thus the standard congruence closure algorithm can be applied to the
simply-typed subset of ITT without much complication. In particular,
we have the familiar property that \lstinline{f a} and \lstinline{g b} are in the same equivalence class if and only if either an equality \lstinline{f a = g b} has been processed, or if \lstinline{f} and \lstinline{g} are in the same equivalence and \lstinline{a} and \lstinline{b} are in the same equivalence class.

\paragraph{Congruence over heterogeneous equality.}
Unfortunately, once we introduce functions with dependent types, we must
switch to \lstinline{heq} and lose the familiar property discussed above that \lstinline{eq}
satisfies for simply-typed functions. Ideally we would like the
following congruence lemma for heterogeneous equality:

\begin{lstlisting}
hcongr_ideal : $\PI$ (A A$'$ : Type) (B : A → Type) (B$'$ : A$'$ → Type)
    (f : $\PI$ (a : A), B a) (f$'$ : $\PI$ (a$'$ : A$'$), B$'$ a$'$) (a : A) (a$'$ : A$'$),
    f == f$'$ $\rightarrow$ a == a$'$ $\rightarrow$ f a == f$'$ a$'$
\end{lstlisting}

Unfortunately, this theorem is not provable in ITT~\cite{hcongr}, even when we assume UIP. The issue is that we need to
establish that \lstinline{B = B$'$} as well, and this fact does not follow from
\lstinline{($\PI$ (a : A), B a) = ($\PI$ (a$'$ : A$'$), B$'$ a$'$)}. Assuming
\lstinline{hcongr_ideal} as an axiom is not a satisfactory solution because it
would limit the applicability of our approach, since as far as we know it is not assumed in any
existing interactive theorem provers based on ITT.

However, for any given $n$, it is straightforward to prove the following congruence lemma
using only \lstinline{erec}, \lstinline{ofheq}
and \lstinline{hrefl}\footnote{The formal statements and proofs for small values of $n$ can be found at\\{\scriptsize\url{http://leanprover.github.io/ijcar16/congr.lean}}, along with formal proofs of all other lemmas described in this paper.}:
\begin{lstlisting}
lemma hcongr$_n$
  (A₁: Type)
  (A₂: A₁ → Type)
  $\ldots$
  (A$_n$: $\PI$ a₁ $\ldots$ a$_{n-2}$, A$_{n-1}$ a₁ $\ldots$ a$_{n-2}$ → Type)
  (B: $\PI$ a₁ $\ldots$ a$_{n-1}$, A$_n$ a₁ $\ldots$ a$_{n-1}$ → Type) :
  $\PI$ (f g: $\PI$ a₁ $\ldots$ a$_{n}$, B a₁ $\ldots$ a$_{n}$), f = g →
  $\PI$ (a₁ b₁: A₁), a₁ == b₁ →
  $\PI$ (a₂: A₂ a₁) (b₂: A₂ b₁), a₂ == b₂ →
  $\ldots$
  $\PI$ (a$_n$: A$_n$ a₁ $\ldots$ a$_{n-1}$) (b$_n$ : A$_n$ b₁ $\ldots$ b$_{n-1}$), a$_n$ == b$_n$ →
  f a₁ $\ldots$ a$_n$ == g b₁ $\ldots$ b$_n$
\end{lstlisting}

The lemmas \lstinline{hcongr$_n$} are weaker than \lstinline{hcongr_ideal} because
they require the outermost functions \lstinline{f} and \lstinline{g} to have the same type.
Although we no longer have the property that \lstinline{f == g} and \lstinline{a == b}
implies \lstinline{f a == g b}, we show in the next section how to extend
the congruence closure algorithm to deal with the additional restriction imposed by \lstinline{hcongr$_n$}.

When using \lstinline{hcongr$_n$} lemmas, we omit the parameters \lstinline{A$_i$}, \lstinline{B}, \lstinline{a$_i$} and
\lstinline{b$_i$} since they can be inferred from the parameters with types
\lstinline{f = g} and \lstinline{a$_i$ == b$_i$}. Note that even if some arguments
of an $n$-ary function \lstinline{f} do not depend on all previous ones, it
is still straightforward to find parameters \lstinline{A$_i$} and \lstinline{B}
that do depend on all previous arguments and so fit the theorem, and
yet become definitionally equal to the types of the actual arguments
of \lstinline{f} once applied to the preceding arguments. We remark that we
avoid this issue in our implementation by synthesizing
custom congruence theorems for every function we encounter.

\section{Congruence Closure\label{sec:cc}}

We now have all the necessary ingredients to describe a very general congruence closure procedure
for ITT. Our procedure is based on the one proposed by Nieuwenhuis and
Oliveras~\cite{CC2005} for first-order logic, which is efficient, is proof producing,
and is used by many SMT solvers. We assume the input to our congruence closure procedure is
of the form \lstinline{$\Gamma$ $\vdash$ a == b}, where $\Gamma$ is a context and \lstinline{a == b} is the goal.
Note that a goal of the form \lstinline{a = b} can be converted into \lstinline{a == b} before we start our procedure, since
when \lstinline{a} and \lstinline{b} have the same type, any proof for \lstinline{a == b} can be converted into a proof for \lstinline{a = b} using
\lstinline{ofheq}. Similarly, any hypothesis of the form \lstinline{e: a = b} can
be replaced with \lstinline{e: a == b} using \lstinline{ofeq}.
As in abstract congruence closure~\cite{Kap97,abstractcongruence}, we introduce new variables
\lstinline{c} to name all proper subterms of every term appearing on either side of an equality,
both to simplify the presentation and to obtain the efficiency of DAG-based
implementations.\footnote{To simplify the presentation further, we ignore the possibility that any of these subterms themselves include partial applications of equality.}
For example, we encode \lstinline{f N a == f N b} using the local definitions
\lstinline{(c$_1$ := f N)} \lstinline{(c$_2$ := c$_1$ a)} \lstinline{(c$_3$ := c$_1$ b)} and
the equality \lstinline{c$_2$ == c$_3$}.
We remark that \lstinline{c$_2$ == c$_3$} is definitionally equal to \lstinline{f N a == f N b}
by $\zeta$-reduction.
Here is an example problem instance for our procedure:
\begin{lstlisting}
(N: Type) (a b: N) (f: $\PI$ A: Type, A $\rightarrow$ A) (c$_1$ := f N)
(c$_2$ := c$_1$ a) (c$_3$ := c$_1$ b) (e: a == b) $\vdash$ c$_2$ == c$_3$
\end{lstlisting}
The term \lstinline{(hcongr$_2$ (refl f) (hrefl N) e)} is a proof for the
goal \lstinline{c$_2$ == c$_3$}.

As in most congruence closure procedures, ours maintains a union-find data
structure that partitions the set of terms into a number of disjoint
subsets such that if \lstinline{a} and \lstinline{b} are in the same subset
(denoted \lstinline{a $\approx$ b}) then the procedure can generate a proof
that \lstinline{a == b}. Each subset is an \emph{equivalence class}.
The union-find data structure computes the equivalence closure of the
relation \lstinline{==} by merging the equivalence classes of \lstinline{a} and
\lstinline{b} whenever \lstinline{e: a == b} is asserted. However, the union-find
data structure alone does not know anything about congruence, and in
particular it will not automatically propagate the assertion \lstinline{a == b}
to other terms that contain \lstinline{a} or \lstinline{b};
for example, it would not merge the equivalence classes of
\lstinline{c := f a} and \lstinline{d := f b}. Thus, additional machinery
is required to find and propagate new equivalences implied by the rules of
congruence.

We say that two terms are \emph{congruent} if they
can be proved to be equivalent using a congruence rule given the
current partition of the union-find data structure.
We also say two local definitions \lstinline{c := f a} and \lstinline{d := g b} are congruent
whenever \lstinline{f a} and \lstinline{g b} are congruent.
We remark that
congruence closure algorithms can be parameterized by the structure of
the congruence rules they propagate. In our case, we use the family
of \lstinline{hcongr$_n$} lemmas as congruence rules.

We now describe our
congruence closure procedure in full, although the overall structure
is similar to the one presented in~\cite{CC2005}. The key
differences are in how we determine whether two terms are congruent, how we build formal proofs of congruence using \lstinline{hcongr$_n$}, and
what local definitions we need to visit after merging two equivalence classes
to ensure that all new congruences are detected. The basic data structures in our procedure are
\begin{itemize}
\item \emph{repr}: a mapping from variables to variables, where $repr[x]$ is
the representative for the equivalence class $x$ is in. We say variable $x$ is a \emph{representative}
if and only if $repr[x]$ is $x$.
\item $next$: a mapping from variables to variables that induces a
  circular list for each equivalence class, where $next[x]$ is the
  next element in the equivalence class $x$ is in.
\item $pr$: a mapping from variables to pairs consisting of a variable
  and a proof, where if $pr[x]$ is $(y, p)$,
then $p$ is a proof for $x == y$ or $y == x$. We use $target[x]$ to denote $pr[x].1$.
This structure implements the \emph{proof forests} described in~\cite{CC2005}.
\item $size$: a mapping from representatives to natural numbers, where
for each representative $x$, $size[x]$ is the number of elements in the equivalence class
represented by $x$.
\item $pending$: a list of local definitions and typing assumptions to be processed.
It is initialized with the context $\Gamma$.
\item $congrtable$: a set of local definitions such that given a local definition
  $E$, the function $lookup(E)$ returns a local definition in
  $congrtable$ congruent to $E$ if one exists.
\item $uselists$: a mapping from representatives to sets of
  local definitions, such that local definition $D$ is in $uselists[x]$ if
  $D$ might become congruent to another definition if the
  equivalence class of $x$ were merged with another equivalence
  class.
\end{itemize}

Our procedure maintains the following invariants for the data structures described above.
\begin{enumerate}
\item $repr[next[x]] \equiv repr[repr[x]] \equiv repr[x]$
\item If $repr[x] \equiv repr[y]$, then $next^k[x] \equiv y$ for some $k$.
\item $target^k[x] \equiv repr[x]$ for some $k$. That is, we can view $target^k[x]$ as
a ``path'' from $x$ to $repr[x]$. Moreover, the proofs in $pr$ can be used to build a
proof from $x$ to any element along this path.
\item Let $s$ be $size[repr[x]]$, then $next^s[x] \equiv x$. That is,
  $next$ does indeed induce a set of disjoint circular lists, one for
  each equivalence class.
\end{enumerate}

Whenever a new congruence proof for \lstinline{c == d} is inferred by our
procedure, we add the auxiliary local definition \lstinline{e: c == d := p} to $pending$, where
\lstinline{e} is a fresh variable, and \lstinline{p} is a proof for \lstinline{c == d}. The proof \lstinline{p} is always
an application of the lemma \lstinline{hcongr$_n$} for some $n$.
We say \lstinline{e : c == d} and \lstinline{e: c == d := p} are \emph{equality proofs} for
\lstinline{c == d}. Given an equality proof $E$, the functions $lhs(E)$ and $rhs(E)$ return
the left and right hand sides of the proved equality.
Given a local definition $E$ of the form \lstinline{c := f a}, the function
$var(E)$ returns \lstinline{c}, and $app(E)$ the pair \lstinline{(f, a)}.
We say a variable \lstinline{c} is a local definition when $\Gamma$ contains the definition \lstinline{c := f a},
and the auxiliary partial function $\mathit{def}(\ttt{c})$ returns this local definition.

\paragraph{Implementing $congrtable$.}
In order to implement the congruence closure procedure
efficiently, the congruence rules must admit a data structure
$congrtable$ that takes a local definition and quickly returns a
local definition in the table that it is congruent to if one exists.
It is easy to implement such a data structure with a Boolean procedure
\Call{congruent}{$D$, $E$} that determines if two local definitions
are congruent, along with a congruence-respecting hash function.
Although the family of \lstinline{hcongr$_n$} lemmas does
not satisfy the property that \lstinline{f a} and \lstinline{g b}
are congruent whenever \lstinline{f $\approx$ g} and
\lstinline{a $\approx$ b},
we still have a straightforward criterion for
determining whether two terms are congruent.

\begin{prop}\label{prop:congruent}
Consider the terms \lstinline{f a} and \lstinline{g b}. If \lstinline{a $\approx$ b}, then
\lstinline{f a} and \lstinline{g b} are congruent provided either:
\begin{enumerate}
\item \lstinline{f} and \lstinline{g} are homogeneously equal;
\item \lstinline{f} and \lstinline{g} are congruent.
\end{enumerate}
\end{prop}
\begin{proof}
  First note that in both cases, we can generate a proof that
  \lstinline{a == b} since we have assumed that
  \lstinline{a $\approx$ b}.
  In the first case, if \lstinline{f} and \lstinline{g} are homogeneously equal, then no
      matter how many partial applications they contain, we can apply
      \lstinline{hcongr$_1$} to the proof of homogeneous equality and
      the proof that \lstinline{a == b}.
      In the second case, if \lstinline{f} and \lstinline{g} are congruent,
      it means that we can generate proofs of all the
      preconditions of \lstinline{hcongr$_k$} for some $k$, and the
      only additional precondition to \lstinline{hcongr$_{k+1}$} is a proof
      that \lstinline{a == b}, which we can generate as well.
\end{proof}

\begin{figure}
\vspace{-14pt}
\begin{algorithmic}[1]
\Procedure{congruent}{$D$, $E$}
  \State $(f, a) \gets app(D)$; $(g, b) \gets app(E)$
  \State \Return $a \approx b$ \textbf{and} \\
          \hspace{36pt} [($f \approx g$ \textbf{and} $type(f) \equiv type(g)$) \textbf{or} \label{line:congr_base} \\
          \hspace{38.5pt}  ($f$ and $g$ are local definitions \textbf{and} \Call{congruent}{$\mathit{def}(f)$, $\mathit{def}(g)$})]
\EndProcedure
\Procedure{congrhash}{$D$}
\State \textbf{given: } $h$, a hash function on terms
\State $(f, a) \gets app(D)$
\State \Return $hashcombine(h(repr[f]), h(repr[a]))$
\EndProcedure
\end{algorithmic}
\caption{Implementing $congrtable$\label{proc:congr}}
\end{figure}

Proposition~\ref{prop:congruent} suggests a simple recursive procedure
to detect when two terms are congruent, which we present in
Figure~\ref{proc:congr}. The procedure \Call{congruent}{$D$, $E$}, where $D$ and $E$ are
local definitions of the form \lstinline{c := f a} and \lstinline{d := g b}, returns \lstinline{true}
if a proof for \lstinline{c == d} can be constructed using an \lstinline{hcongr$_n$} lemma for some $n$.
Note that although the congruence lemmas \lstinline{hcongr$_n$} are
themselves $n$-ary, it is not sufficient to view the two terms being
compared for congruence as applications of $n$-ary functions. We
must compare each pair of partial applications for homogeneous
equality as well~(line~\ref{line:congr_base}), since two terms
with $n$ arguments each might be congruent using \lstinline{hcongr$_m$} for any $m$ such that $m \leq
n$. For example, \lstinline{f a1 c} and \lstinline{g b1 c} are
congruent by \lstinline{hcongr$_2$} if \lstinline{f = g} and
\lstinline{a1 == b1}, and yet are only congruent by
\lstinline{hcongr$_1$} if all we know is \lstinline{f a1 = g b1}. It
is even possible for two terms to be congruent that do
not have the same number of arguments. For example, \lstinline{f = g a}
implies that \lstinline{f b} and \lstinline{g a b} are congruent by \lstinline{hcongr$_1$}.

Proposition~\ref{prop:congruent} also suggests a simple way to
hash local definitions that respects congruence. Given a hash
function on terms, the procedure \Call{congrhash}{$D$} hashes a
local definition of the form \lstinline{c := f a} by simply combining the hashes
of the representatives of \lstinline{f} and \lstinline{a}. This hash function
respects congruence because if \lstinline{c := f a} and \lstinline{d := g b} are
congruent, it is a necessary (though not sufficient) condition that
\lstinline{f $\approx$ g} and \lstinline{a $\approx$ b}.

\begin{figure}
\begin{algorithmic}[1]
\Procedure{cc}{$\Gamma \vdash a == b$}
\State $pending \gets \Gamma$
\While{$pending$ is not empty}
   \State remove next $E$ from $pending$
   \If{$E$ is an equality proof}
      \Call{processeq}{$E$}
   \Else\
      \Call{initialize}{$E$}
   \EndIf
\EndWhile
\If{$repr[a] \equiv repr[b]$}
   \Return \Call{mkpr}{a, b}
\Else\ \textbf{fail} \EndIf
\EndProcedure
\end{algorithmic}
\caption{Congruence closure procedure\label{proc:cc}}
\end{figure}

\paragraph{The procedure.} Figure~\ref{proc:cc} contains the main procedure \Call{cc}{}.
It initializes $pending$ with the input context $\Gamma$.
Variables in typing assumptions and local definitions are processed using
 \Call{initialize}{} (Figure~\ref{proc:init}), and equality
proofs are processed using \Call{processeq}{}
(Figure~\ref{proc:proceq}).

\begin{figure}
\begin{algorithmic}[1]
\Procedure{initialize}{$E$}
  \State $c \gets var(E)$
  \State $repr[c] \gets c$; $next[c] \gets c$; $size[c] \gets 1$; $uselists[c] \gets \emptyset$
  \State $pr[c] \gets (c, $\ttt{`hrefl $c$'}$)$
  \If{$E$ is a local definition}
     \State \Call{inituselist}{$E$, $E$}
     \If{$D = lookup(E)$}
        \State $d \gets var(D)$; $e \gets$ make fresh variable
        \State add ($e : d == c$ := \Call{mkcongr}{$D$, $E$, []}) to $pending$ and $\Gamma$
     \Else\
        add $E$ to $congrtable$
     \EndIf
   \EndIf
\EndProcedure
\Procedure{inituselist}{$E$, $P$}
  \State $(f, a) \gets app(E)$
  \State add $P$ to $uselists[f]$ and $uselists[a]$
  \If{$f$ is a local definition}
     \Call{inituselist}{$\mathit{def}(f)$, $P$}
  \EndIf
\EndProcedure
\end{algorithmic}
\caption{Initialization procedure\label{proc:init}}
\end{figure}

The \Call{initialize}{$E$} procedure invokes \Call{inituselist}{$E$, $E$}
whenever $E$ is a local definition \lstinline{$c$ := $f$ $a$}. The second argument
at \Call{inituselist}{$E$, $P$} represents the \emph{parent} local definition
that must be included in the $uselists$.
We must ensure that for every local definition $D$ that could
be inspected during a call to \Call{congruent}{$E_1$,~$E_2$} for some $E_2$, we add
$var(E_1)$ to the $uselist$ of $var(D)$ when initializing $E_1$. Thus
the recursion in \Call{inituselist}{} must mirror the recursion in
\Call{congruent}{} conservatively, and always recurse whenever
\Call{congruent}{} might recurse. For example, assume the input context $\Gamma$ contains
\begin{lstlisting}
(A: Type) (a b d: A) (g : A $\rightarrow$ A $\rightarrow$ A) (f : A $\rightarrow$ A) (c$_1$ := g a) (c$_2$ := c$_1$ b) (c$_3$ := f d).
\end{lstlisting}
When \Call{initialize}{\ttt{c$_2$ := c$_1$ b}} is invoked, \lstinline{c$_2$ := c$_1$ d} is added to the $uselists$ of
\lstinline{c$_1$}, \lstinline{b}, \lstinline{g} and \lstinline{a}.
By a slight abuse of notation, we write \lstinline{`hrefl $a$'} to represent in the pseudocode
the expression that creates the \lstinline{hrefl}-application
using as argument the term stored in the program variable $a$.

The procedure \Call{processeq}{} is used to process equality proofs
\lstinline{a == b}.
If \mbox{\lstinline{a}} and \lstinline{b} are already in the same
equivalence class, it does nothing. Otherwise, it first removes every element in $uselists[repr[a]]$ from $congrtable$ (procedure \Call{removeuses}{}). Then, it merges the
equivalence classes of $a$ and $b$ so that
for every $a'$ in the equivalence class of $a$, $repr[a']$ is set to $repr[b]$.
This operation can be implemented efficiently using the $next$ data structure.
As in~\cite{CC2005}, the procedure also reorients the path from $a$ to $repr[a]$
induced by $pr$ (procedure \Call{flipproofs}{}) to make sure invariant 3 is still satisfied
and \emph{locally irredundant transitivity proofs}~\cite{de2005justifying} can be generated.
It then reinserts the elements removed by \Call{removeuses}{} into $congrtable$ (procedure \Call{reinsertuses}{}); if any are found to be congruent to an existing term in a different partition,
it proves equivalence using the congruence lemma \lstinline{hcongr$_n$}
(procedure \Call{mkcongr}{}) and puts the new proof
onto the queue. Finally, \Call{processeq}{} updates $next$, $uselists$ and $size$ data structures.

\begin{figure}
\begin{algorithmic}[1]
\Procedure{processeq}{$E$}
   \State $a \gets lhs(E);\ b \gets rhs(E)$
   \If{$repr[a] \equiv repr[b]$}
      \Return
   \EndIf
   \If{$size(repr[a]) > size(repr[b])$}
      \textbf{swap}$(a, b)$
   \EndIf
   \State $r_a \gets repr[a]$; $r_b \gets repr[b]$
   \State \Call{removeuses}{$r_a$}; \Call{flipproofs}{$a$}
   \ForAll{$a'$ s.t. $repr[a'] \equiv r_a$} $repr[a'] \gets r_b$ \EndFor
   \State $pr[a] \gets (b, E)$
   \State \Call{reinsertuses}{$r_a$}
   \State \textbf{swap}$(next[r_a], next[r_b])$
   \State \textbf{move} $uselists[r_a]$ \textbf{to} $uselists[r_b]$; $size[r_b] \gets size[r_b] + size[r_a]$
\EndProcedure
\Procedure{flipproofs}{$a$}
   \If{$repr[a] \equiv a$} \Return \EndIf
   \State $(b, p) \gets pr[a]$; \Call{flipproofs}{$b$}; $pr[b] \gets (a, p)$
\EndProcedure
\Procedure{removeuses}{$a$}
  \ForAll{$E$ in $uselists[a]$} remove $E$ from $congrtable$ \EndFor
\EndProcedure
\Procedure{reinsertuses}{$a$}
  \ForAll{$E$ in $uselists[a]$}
  \If{$D = lookup(E)$}
     \State $d \gets var(D)$; $e \gets var(E)$; $p \gets$ make fresh variable
     \State add ($p : d == e$ := \Call{mkcongr}{$D$, $E$, []}) to $pending$ and $\Gamma$
  \Else\ add $E$ to $congrtable$
  \EndIf
  \EndFor
\EndProcedure
\end{algorithmic}
\caption{Process equality procedure\label{proc:proceq}}
\end{figure}

Figure~\ref{proc:mkpr} contains a simple recursive procedure \Call{mkcongr}{} to construct the proof that two congruent local definitions are
equal.
The procedure takes as input two local definitions $D$ and $E$ of the form \lstinline{c := f a} and \lstinline{d := g b}
such that \Call{congruent}{$D$, $E$}, along with a possibly empty list of equality proofs $es$ for
\lstinline{a$_1$ == b$_1$}, \ldots, \lstinline{a$_n$ == b$_n$}, and returns a proof for
\lstinline{f a a$_1$ ... a$_n$ == g b b$_1$ ... b$_n$}.
The two cases in the \Call{mkcongr}{} procedure mirror the two cases of the \Call{congruent}{} procedure.
If the types of \lstinline{f} and \lstinline{g} are definitionally equal we construct an instance of the lemma \lstinline{hcongr$_{|es|+1}$}.
The procedure \Call{mkpr}{\lstinline{a}, \lstinline{b}} (Figure~\ref{proc:mkpr}) creates a proof for \lstinline{a == b} if \lstinline{a} and \lstinline{b} are in the same equivalence class
by finding the common element $target^n[\ttt{a}] \equiv target^m[\ttt{b}]$ in the ``paths'' from \lstinline{a} and \lstinline{b} to the equivalence class representative.
Note that, if \Call{congruent}{$D$, $E$} is true, then \Call{mkcongr}{$D$, $E$, []} is a proof for \lstinline{c == d}.

\begin{figure}
\begin{algorithmic}[1]
\Procedure{mkcongr}{$D$, $E$, $es$}
  \State \textbf{assumption:} \Call{congruent}{$D$, $E$}
  \State $(f, a) \gets app(D)$; $(g, b) \gets app(E)$; $e_{ab} \gets$ \Call{mkpr}{$a$, $b$}
  \If{$type(f) \equiv type(g)$}
    \State $n \gets len(es)$; $e_{fg} \gets$ \Call{mkpr}{$f$, $g$}
    \State \Return \lstinline{`hcongr$_{n+1}$ (ofheq $e_{fg}$) $e_{ab}$ $es$'}
  \Else\
    \Return \Call{mkcongr}{$\mathit{def}(f)$, $\mathit{def}(g)$, [$es$, $e_{ab}$]}
  \EndIf
\EndProcedure
\Procedure{mkpr}{$a$, $b$}
  \If{$a \equiv b$}
     \Return \lstinline{`hrefl $a$'}
  \EndIf
  \State \textbf{let} $n$ and $m$ be the smallest values s.t. $target^n[a] \equiv target^m[b]$
  \State $e_a \gets$ \Call{mktrans}{$a$, $n$}; $e_b \gets$ \Call{mktrans}{$b$, $m$};
  \Return \lstinline{`htrans $e_a$ (hsymm $e_b$)'}
\EndProcedure
\Procedure{mktrans}{$a$, $n$}
  \If{$n = 0$}
     \Return \lstinline{`hrefl $a$'}
  \EndIf
  \State $(b, e_{ab}) \gets pr[a]$;  $e \gets$ \Call{mktrans}{$b$, $n-1$}
  \If{$lhs(e_{ab}) \equiv a$ \textbf{and} $rhs(e_{ab}) \equiv b$}
      \Return \lstinline{`htrans $e_{ab}$ $e$'}
  \Else\
      \Return \lstinline{`htrans (hsymm $e_{ab}$) $e$'}
  \EndIf
\EndProcedure
\end{algorithmic}
\caption{Transitive proof generation procedure\label{proc:mkpr}}
\end{figure}

Finally, we remark that the main loop of \Call{cc}{} maintains the following two invariants.
\begin{theorem}
If $a$ and $b$ are in the same equivalence class \textup{(}i.e., $a \approx b$\textup{)}, then \Call{mkpr}{}\textup{(}$a$, $b$\textup{)} returns a correct proof that \lstinline{$a$ == $b$}.
\end{theorem}

\begin{theorem}
If $type(f) \equiv type(g)$, $f \approx g$, $a_1 \approx b_1$, \ldots $a_n \approx b_n$, $c \equiv f\ a_1 \ldots a_n$ and $d \equiv g\ b_1 \ldots b_n$, then
$c \approx d$.
\end{theorem}

\paragraph{Extensions.}
There are many standard extensions to the congruence closure procedure
that are straightforward to support in our framework, such as
tracking disequalities to find contradictions and propagating
injectivity and disjointness for inductive datatype
constructors~\cite{fewconstructions}. Here we present a simple
extension for propagating equalities among elements of \emph{subsingleton} types that is especially
important when proving theorems in ITT.
We say a type \lstinline{A:Type} is a subsingleton if it has at
most one element; that is, if for all
\lstinline{(a b:A)}, we have that \lstinline{a = b}. Subsingletons are used extensively in practice, and are especially ubiquitous when \emph{proof irrelevance} is assumed, in which case every proposition is a subsingleton.

One common use of dependent types is to extend functions to take
extra arguments that represent proofs that certain preconditions
hold. For example, the logarithm function only makes sense for
positive real numbers, and we can make it impossible to even call
it on a non-positive number by requiring a proof of positivity as
a second argument: \lstinline{safe_log : $\PI$ x:$\mathbb{R}$, x $>$ 0 $\rightarrow$ $\mathbb{R}$}.
The second argument is a proposition and hence is a subsingleton when
we assume \emph{proof irrelevance}.
Consider the following goal:
\lstinline{(a b : $\mathbb{R}$) (Ha : a $>$ 0) (Hb : b $>$ 0) (e : a = b) $\vdash$ safe_log a Ha = safe_log b Hb}.
The core procedure we presented above would not be able to prove this
theorem on its own because it would never discover that \lstinline{Ha == Hb}.
We show how to extend the procedure to automatically propagate facts of this kind.

We assume we have an oracle $issub(\Gamma, A)$ that returns true for subsingleton types
for which we have a proof \lstinline{sse$_A$} of \lstinline{$\PI$a b:A, a = b}.
Many proof assistants implement an efficient (and incomplete) $issub$ using
\emph{type classes}~\cite{casteran2014gentle,elab2015}, but it is beyond the scope of this paper to describe this mechanism.
Given a subsingleton type \lstinline{A} with proof \lstinline{sse$_A$}, we can prove
\begin{lstlisting}
  hsse$_A$: $\PI$ (C:Type) (c:C) (a:A), C == A $\rightarrow$ c == a,
\end{lstlisting}
which we can use as an additional propagation rule in the congruence closure procedure.
The idea is to merge the equivalence classes of \lstinline{a:A} and \lstinline{c:C}
whenever \lstinline{A} is a subsingleton and \lstinline{C $\approx$ A}.
First, we add a mapping $subrep$ from subsingleton types to their representatives.
Then, we include the following additional code in \Call{initialize}{}:
\begin{algorithmic}
  \State $C  \gets type(c)$; $A \gets repr[C]$
  \If{$issub(\Gamma, A)$}
  \If{$a = subrep[A]$}
  \State $p \gets $ \Call{mkpr}{$C$, $A$}; $e \gets$ make fresh variable
  \State add ($e : c == a$ := \lstinline{hsse$_A$ $C$ $p$ $c$ $a$}) to $pending$ and $\Gamma$
  \Else\ $subrep[A] \gets c$
  \EndIf
  \EndIf
\end{algorithmic}
Finally, at \Call{processeq}{} whenever we merge the equivalence classes of subsingleton
types $A$ and $C$, we also propagate the equality \lstinline{$subrep[A]$ == $subrep[C]$}.

With this extension, our procedure can prove \lstinline{safe_log a Ha = safe_log b Hb} in the example
above, since the terms \lstinline{a $>$ 0} and \lstinline{b $>$ 0} are both
subsingleton types with representative elements \ttt{Ha} and \ttt{Hb}
respectively, and when their equivalence classes are merged, the
subsingleton extension propagates the fact that their representative
elements are equal, i.e. that \lstinline{Ha == Hb}.

\section{Applications\label{sec:applications}}

We have implemented our congruence closure procedure for Lean\footnote{\scriptsize{\url{https://github.com/leanprover/lean/blob/master/src/library/blast/congruence_closure.cpp}}}
along with many of the standard extensions as part of a long-term effort to build a robust theorem prover for
ITT. Although congruence closure can be useful on its own, its power
is greatly enhanced when it is combined with a procedure for
automatically instantiating lemmas so that the user does not need to
manually collect all the ground facts that the congruence closure
procedure will need. We use an approach called
\emph{e-matching}~\cite{ematching} to instantiate lemmas that makes
use of the equivalences represented by the state of the congruence
closure procedure when deciding what to instantiate, though the details of
e-matching are beyond the scope of this paper. The combination of
congruence closure and e-matching is already very powerful, as we
demonstrate in the following two examples, the first from software
verification and the second from formal mathematics. The complete list
of examples we have used to test our procedure can be found at \url{http://leanprover.github.io/ijcar16/examples}.

\paragraph{Vectors (indexed lists).}
As we mentioned in \S\ref{sec:equality}, a common use of dependent
types is to include the length of a list inside its type in order to
make out-of-bounds errors impossible.
The constructors of \lstinline{vector} mirror those of \lstinline{list}:
\begin{lstlisting}
nil : $\PI$ {A : Type}, vector A 0
cons : $\PI$ {A : Type} {n : ℕ}, A $\rightarrow$ vector A n $\rightarrow$ vector A (succ n)
\end{lstlisting}
where \lstinline{succ} is the successor function on natural numbers, and where curly braces indicate that a parameter should be inferred from context. We use the notation \lstinline{[x]} to denote the one-element
\lstinline{vector} containing only \lstinline{x}, i.e. \lstinline{cons x nil}, and
\lstinline{x::v} to denote \lstinline{cons x v}.
It is easy to define append and reverse on \lstinline{vector}:
\begin{lstlisting}
app : $\PI$ {A : Type} {n₁ n₂ : $\mathbb{N}$}, vector A n₁ $\rightarrow$ vector A n₂ $\rightarrow$ vector A (n₁ + n₂)
rev : $\PI$ {n : $\mathbb{N}$}, vector A n $\rightarrow$ vector A n
\end{lstlisting}
When trying to prove the basic property \lstinline{rev (app v₁ v₂) == app (rev v₂) (rev v₁)} about these two functions,
we reach the following goal:
\begin{lstlisting}
(A : Type) (n₁ n₂ : $\mathbb{N}$) (x₁ x₂ : A) (v₁ : vector A n₁) (v₂ : vector A n₂)
(IH : rev (app v₁ (x₂::v₂)) == app (rev (x₂::v₂)) (rev v₁))
$\vdash$ rev (app (x₁::v₁) (x₂::v₂)) == app (rev (x₂::v₂)) (rev (x₁::v₁))
\end{lstlisting}
Given basic lemmas about how to push \ttt{app} and \ttt{rev} in over \ttt{cons},
a lemma stating the associativity of \ttt{app}, and a few
basic lemmas about natural numbers, our congruence closure procedure
together with the e-matcher can solve this goal. Once the e-matcher
establishes the following ground facts:
\begin{lstlisting}
H₁ : rev (x₁::v₁) == app (rev v₁) [x₁]
H₂ : app (x₁::v₁) (x₂::v₂) == x₁::(app v₁ (x₂::v₂))
H₃ : rev (x₁::(app v₁ (x₂::v₂))) == app (rev (app v₁ (x₂::v₂))) [x₁]
H₄ : app (app (rev (x₂::v₂)) (rev v₁)) [x₁] == app (rev (x₂::v₂)) (app (rev v₁) [x₁])
\end{lstlisting}
as well as a few basic facts about the natural numbers, the result
follows by congruence.

\paragraph{Safe arithmetic.}
As we mentioned in \S\ref{sec:cc}, another common use of dependent types
is to extend functions to take extra arguments that represent proofs
that certain preconditions hold. For example, we can define safe
versions of the logarithm function and the inverse function as follows:
\begin{lstlisting}
        safe_log : $\PI$ (x : $\mathbb{R}$), x $>$ 0 $\rightarrow$ $\mathbb{R}$          safe_inv : $\PI$ (x : $\mathbb{R}$), x $\ne$ 0 $\rightarrow$ $\mathbb{R}$
\end{lstlisting}
Although it would be prohibitively cumbersome to prove the
preconditions manually at every invocation, we can relegate this task
to the theorem prover, so that \lstinline{log x} means \lstinline{safe_log x p}
and \lstinline{y$^{-1}$} means \lstinline{safe_inv y q}, where \lstinline{p} and \lstinline{q}
are proved automatically. Given basic lemmas about arithmetic
identities, our congruence closure procedure together with the
e-matcher can solve many complex equational goals like the following,
despite the presence of embedded proofs:
\begin{lstlisting}
$\forall$ (x y z w : $\mathbb{R}$), x > 0 $\to$ y > 0 $\to$ z > 0 $\to$ w > 0 $\to$ x * y = exp z + w $\to$
  log (2 * w * exp z + w$^2$ + exp (2 * z)) / -2 = log y$^{-1}$ - log x
\end{lstlisting}

\section{Related Work\label{sec:related_work}}

Corbineau~\cite{corbineau2001autour} presents a
congruence closure procedure for the simply-typed subset of ITT and a
corresponding implementation for Coq as the tactic \lstinline{congruence}. The procedure uses
homogeneous equality and does not support dependent types at
all. Hur~\cite{hur2010heq} presents a library of
tactics for reasoning over a different variant of heterogeneous
equality in Coq, for which the user must manually separate the parts of the type
that are allowed to vary between heterogeneously equal terms from those that
must remain the same. The main tactic provided is \lstinline{Hrewritec}, which tries
to rewrite with a heterogeneous equality by converting it to a cast-equality,
rewriting with that, and then generalizing the proof that the types are
equal. There does not seem to be any general notion of congruence akin to our
family of \lstinline{hcongr$_n$} lemmas.

Sj\"oberg and Weirich~\cite{sjoberg2014programming} propose using
congruence closure during type checking for a new dependent type
theory in which definitional equality is determined by the congruence
closure relation instead of by the standard forms of reduction.  They
avoid the problem that we solve in this paper by designing their type
theory so that \lstinline{hcongr_ideal} is provable. However, since
their type theory is not compatible with any of the standard flavors
of ITT such as the calculus of inductive constructions, their
congruence closure procedure cannot be used to prove theorems in
systems such as Coq and Lean.

\section{Conclusion}

We have presented a very general notion of congruence for ITT based on
heterogeneous equality that applies to all dependently typed
functions. We also presented a congruence closure procedure that can
propagate the associated congruence rules efficiently and so
automatically prove a large and important set of goals.
Just as congruence closure procedures (along with DPLL) form the foundation of modern SMT
solvers, we hope that our congruence closure procedure can form the foundation
of a robust theorem prover for intensional type theory. We are building
such a theorem prover for Lean, and it can already solve many interesting problems.

\paragraph{Acknowledgments.}
We would like to thank David Dill, Jeremy Avigad, Robert Lewis, Nikhil Swany, Floris van Doorn and Georges Gonthier for
providing valuable feedback on early drafts.

\bibliographystyle{splncs03}
\bibliography{congr}

\end{document}